\documentclass[final,3p,times]{elsarticle} 

\usepackage{amssymb}
\usepackage{amsthm}

\usepackage{bm}
\usepackage{url}
\usepackage{amsmath}

\newtheorem{theorem}{Theorem}

\newcommand{\R}{\mathbb{R}}
\newcommand{\Rn}{\R^{n}}
\newcommand{\Rnn}{\R^{n \times n}}
\newcommand{\abs}[1]{|#1|}
\newcommand{\ap}[1]{\hat{#1}}
\newcommand{\norm}[1]{\|#1\|}
\newcommand{\normtwo}[1]{\|#1\|_{2}}
\newcommand{\norminf}[1]{\|#1\|_{\infty}}
\newcommand{\diag}{\mathrm{diag}}
\newcommand{\trans}[1]{#1^{\top}}

\journal{Journal of Computational and Applied Mathematics}

\begin{document}

\begin{frontmatter}



\title{An a posteriori verification method for generalized real-symmetric eigenvalue problems in large-scale electronic state calculations}


\author[Tottori_U]{Takeo Hoshi}
\author[TWCU]{Takeshi Ogita}
\author[SIT1]{Katsuhisa Ozaki}
\author[SIT2]{Takeshi Terao}

\address[Tottori_U]{Department of Applied Mathematics and Physics, Tottori University, Japan}
\address[TWCU]{Division of Mathematical Sciences, Tokyo Woman's Christian University, Japan}
\address[SIT1]{Department of Mathematical Sciences, Shibaura Institute of Technology, Japan}
\address[SIT2]{Graduate School of Engineering and Science, Shibaura Institute of Technology, Japan}

\begin{abstract}
An a posteriori verification method is proposed for the generalized real-symmetric eigenvalue problem
and is applied to densely clustered eigenvalue problems
in large-scale electronic state calculations. 
The proposed method is realized by a two-stage process in which the approximate solution is computed by existing numerical libraries and is then verified in a moderate computational time.
The procedure returns intervals containing one exact eigenvalue in each interval. 
Test calculations were carried out for organic device materials, and 
the verification method confirms that all exact eigenvalues are well separated in the obtained intervals.
This verification method will be integrated into
EigenKernel (\url{https://github.com/eigenkernel/}), 
which is middleware for various parallel solvers for the generalized eigenvalue problem.
Such an a posteriori verification method will be important in future computational science.
\\
(c) 2005 Elsevier B.V. All rights reserved.
\end{abstract}


\end{frontmatter}



\section{Introduction \label{SEC-INTRO}}

A crucial issue in verification methods is application to large-scale scientific or industrial computations on supercomputers. 
Many numerical solvers have been proposed for modern massively parallel supercomputers,
and application researchers would like to compare solvers both in terms of computational speed
and reliability. The concept of a posteriori verification methods is proposed in order to meet the needs of application researchers.

A posteriori verification methods have the workflow shown in Fig.~\ref{FIG-SCHEMATIC}.
An approximate solution is first obtained and then verified. 
The former and latter procedures are referred to as a solver and a verifier, respectively.
The goal of the present study is to integrate the verifier routine, as an optional function, to existing numerical solver libraries.

\begin{figure}[h]
\begin{center}
  \includegraphics[width=0.6\textwidth]{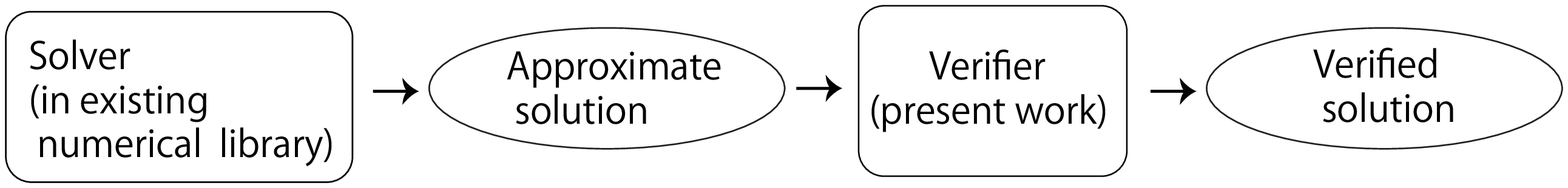}
\end{center}
\caption{Schematic diagram of the workflow with an a posteriori verification method. 
}
\label{FIG-SCHEMATIC}       
\end{figure}

The present research is motivated by large-scale electronic state calculation,
a major field in computational material science and engineering.
As explained in \ref{SEC-GHEV}, a mathematical model is used
for the fundamental Schr\"{o}dinger-type equation, and the problem is reduced to 
the generalized real-symmetric matrix eigenvalue problem 
\begin{eqnarray}
 A x_k = \lambda_k B x_k
 \label{EQ-QM-GEP}
\end{eqnarray}
under the generalized orthogonality condition  
\begin{eqnarray}
 x_i^{\rm T} B   x_j = \left\{\begin{array}{ll} 1 & \mathrm{if} \ i = j \\ 0 & \mathrm{otherwise} \end{array}\right.,
 \label{EQ-ORTHO-CON}
\end{eqnarray}
where both $A$ and $B$ are real-symmetric $n \times n$ matrices, with $B$ being positive definite. 
Here, we assume that
\[
\lambda_1 \le \lambda_2 \le \dots \le \lambda_n .
\]
Applying our results to problems with complex Hermitian matrices is straightforward. 

For large-scale electronic state calculation,
many eigenvalues are densely clustered or almost degenerate,  
and distinguishing them numerically may be difficult. 
In order to obtain reliable results, we consider verification methods for generalized eigenvalue problems.
For the sake of completeness as verification methods, we also need to take into account all numerical errors that occur when matrices $A$ and $B$ are generated from the fundamental Schr\"{o}dinger-type equation.
Although we do not consider the fundamental Schr\"{o}dinger-type equation in detail herein, we briefly discuss this equation in \ref{SEC-GHEV}.

One of the authors (T. H.) developed 
a middleware EigenKernel \cite{EIGENKERNEL-URL, IMACHI2016-JIP, 2018TANAKA} 
with various parallel solvers for generalized eigenvalue problems and plans to add a verifier routine. 
The total elapsed time $T_{\rm tot}$ is 
the sum of the times for solver 
$T_{\rm sol}$ and verifier $T_{\rm veri}$
($T_{\rm tot} = T_{\rm sol} + T_{\rm veri}$).
We attempt to construct the verifier algorithm
so that the time for the verifier gives a moderate fraction ($T_{\rm veri} \le T_{\rm sol}$). 
Since the verifier can use the highly optimized routines of matrix multiplication, the verifier is suitable for high-performance computing on supercomputers.

In the solver procedure,
approximate solutions $(\ap{\lambda}_k, \ap{x}_k)$, $k = 1, 2, \dots, n$, such that
\begin{eqnarray}
 A \ap{x}_k \approx \ap{\lambda}_k B \ap{x}_k
 \label{EQ-QM-GEP-APPROX}
\end{eqnarray}
are obtained by any numerical solver algorithm.
A verifier procedure gives the difference between the exact and approximate solutions, 
such as $|\lambda_k - \ap{\lambda}_k|$ or 
$\norm{x_k - \ap{x}_k}$. 
If the relation 
$|\lambda_k - \ap{\lambda}_k| \le r_k$ is obtained with a given positive number $r_k$, for example,
this indicates that the exact solution ($\lambda_k$) lies 
in a disk having a center and radius of $\ap{\lambda}_k$ and $r_k$, respectively. 
For this purpose, a number of enclosure methods have been developed, e.g., \cite{Ya2001,MiOgRuOi2010,Mi2012}.
In the present paper, we propose a method of enclosing all eigenvalues that is straightforward, efficient, and easy to implement on supercomputers.
The proposed method is based on Yamamoto's theorem \cite{Ya1984} and is essentially the same as the method proposed in a previous paper \cite{Mi2012}.
In other words, we specialize the previous method \cite{Mi2012} for generalized real-symmetric eigenvalue problems.
Note that it is not possible in general to state that a method is better or worse than other methods because this depends on the purpose.
We compare the advantages and disadvantages of these enclosure methods in Section~\ref{SEC-VERIF-METHODS}.

The a posteriori verification strategy is important mainly with regards to three aspects. 
First, numerical methods for the densely clustered eigenvalue problem 
have potential difficulties in computing reliable numerical solutions,
as explained above. 
Second, 
various numerical algorithms have been proposed for efficient parallel computations that are suitable for current and next-generation supercomputers.
Application researchers would like to compare these methods 
with respect to both computational speed and numerical reliability.  
Third, 
the emergence of machine learning has enhanced 
the design of computer architecture  
for the acceleration of low-precision (single- or half-precision) calculation. 
The efficient use of low-precision calculation,  
typically in mixed-precision calculation, 
will be important in any high-performance computational science field \cite{DONGARRA2018-HPCASIA, ALVERMANN2018}.
A posteriori verification methods guarantee satisfactory numerical reliability when low-precision calculation is used.

The remainder of the present paper is organized as follows.
Section ~\ref{SEC-BACKGROUND} explains 
the physical and mathematical backgrounds.
The proposed verification method and 
numerical examples
are presented in Sections \ref{SEC-VERIF-METHODS} and \ref{SEC-NUM-EXAMPLE}, respectively. 
Section \ref{SEC-SUMMARY} presents a summary and an outlook for future research.

\section{Background \label{SEC-BACKGROUND}}

\subsection{Large-scale electronic state calculation and densely clustered eigenvalue problem  \label{SEC-CLUSTERED}}

The present electronic state calculation is briefly introduced in \ref{SEC-GHEV}. 
The matrix size $n$ is approximately proportional to the number of the atoms, molecules, or electrons in the material. 
An eigenvalue $(\lambda_k)$ or an eigenvector  $(x_k)$ indicates the energy and the wavefunction, respectively, of an electron.

The present research is motivated, in particular, 
by a previous study \cite{HOSHI2018-PENTA},
in which we focused on the participation ratio \cite{BELL1970,1996FUJIWARA}
defined for a vector $v \equiv (v_1, v_2, ...., v_n)$, as
\begin{eqnarray}
 P \equiv P(v) = \left( \sum_{j=1}^n |v_j|^4 \right)^{-1}.
 \end{eqnarray}
The participation ratio is a measure of the spatial extension of the electronic wavefunction
and governs the electronic device properties. A dense eigenvector, {\it i.e.}, a vector that has only a few components that are negligible in terms of absolute value, has a large participation ratio.
The corresponding electronic wavefunction is extended through the material
and can contribute to electrical current. 
A sparse eigenvector, {\it i.e.}, a vector that has only a few components that are large in terms of absolute value, indicates a small participation ratio.
The corresponding electronic wavefunction is localized in the material
and cannot contribute to electrical current.
An interesting research target in a large-scale problem is
an \lq intermediate' electronic wavefunction or a wavefunction that 
shows intermediate properties between extended and localized wavefunctions.
Such \lq intermediate' wavefunctions appear, for example, in
Fig. 1 of Ref.~\cite{1996FUJIWARA} or Fig. 3 of Ref.~\cite{HOSHI2018-PENTA}.

The densely clustered eigenvalue problem in (\ref{EQ-QM-GEP}) appears among large-scale calculations
and is illustrated in Fig.~\ref{FIG-CLUSTERED}.
In the problem, 
the difference of sequential eigenvalues $\delta_k \equiv \lambda_{k+1} - \lambda_{k}$, $k = 1, 2, \dots, n-1$,
tends to be proportional to $1/n$ $(\delta_k  \propto 1/n)$.
Consequently, many eigenvalues are densely clustered or almost degenerate 
($\delta_k \rightarrow 0$)
in a large-matrix problem $(n \rightarrow \infty)$ 
and distinguishing these eigenvalues numerically may be difficult. 

It is crucial to distinguish each eigenvalue numerically among densely clustered eigenvalues,
because the participation ratio and other physical quantities are defined for each eigenvector. 
If two calculated 
eigenvalues $\hat{\lambda}_{k}$ and $\hat{\lambda}_{k+1}$ cannot be distinguished in the numerical calculation,
or if the two eigenvalues are recognized, unphysically, to be degenerate, 
then the corresponding eigenvectors $\hat{x}_k$ and $\hat{x}_{k+1}$ cannot be defined uniquely. 
In this case, the participation ratio values $P(\hat{x}_k)$ and $P(\hat{x}_{k+1})$ are not defined uniquely and any discussion of these values will be meaningless.

\begin{figure}[h]
\begin{center}
  \includegraphics[width=0.4\textwidth]{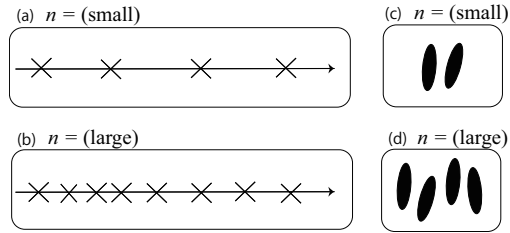}
\end{center}
\caption{Schematic diagram of a densely clustered eigenvalue problem in large-scale electronic state calculations. 
(a), (b) Eigenvalue distribution in (\ref{EQ-QM-GEP}) with (a) small or (b) large matrix size $n$. 
A cross indicates an eigenvalue on the real axis. 
The difference of sequential eigenvalues tends to be proportional to $1/n$.
(c), (d) Materials with (c) small or (d) matrix size $n$.
Ovals indicate molecules. 
The matrix size $n$ is proportional to the size of the molecules $n_{\rm mol}$ 
($n \propto n_{\rm mol}$). 
}
\label{FIG-CLUSTERED}       
\end{figure}


\subsection{Numerical solvers for the generalized eigenvalue problem \label{SEC-SOLVER-GEV}}

Here, an overview is given for 
the parallel dense-matrix solver of the generalized eigenvalue problem of (\ref{EQ-QM-GEP}), 
in particular, for the variety of used algorithms. 
The solver algorithm for (\ref{EQ-QM-GEP}) consists of four procedures: 
(i) Cholesky decomposition of $B$ 
\begin{eqnarray}
B=R^{\top}R, 
\end{eqnarray}
with an upper triangular matrix $R$, 
(ii) reduction to the standard eigenvalue problem (SEP)
\begin{eqnarray}
A^{\prime}y_k=\lambda_k y_k, 
\label{EQ-RED-SEP}
\end{eqnarray}
with 
\begin{eqnarray}
A^{\prime} \equiv R^{-\top}AR^{-1}, 
\end{eqnarray}
(iii) solution of the standard eigenvalue problem (\ref{EQ-RED-SEP}),
and (iv) transformation of the eigenvectors 
\begin{eqnarray}
x_k = R^{-1}y_k. 
\end{eqnarray}
The set of procedures (i), (ii), and (iv) is referred to as 
reducer, and procedure (iii) is referred to as the SEP solver.

Although ScaLAPACK \cite{SCALAPACK-URL, SCALAPACK-BOOK} is the {\it de facto} standard parallel numerical library, 
this library was developed mainly in the 1990s, and several routines exhibit severe bottlenecks 
on modern massively parallel supercomputers. 
Novel solver libraries of ELPA \cite{ELPA-URL, ELPA2014} and EigenExa \cite{EIGENEXA-URL,EigenExa-PAPER} were proposed in order to overcome the bottlenecks. 
The ELPA code was developed in Europe under tight collaboration
between computer scientists and material science researchers, and
its main target application is FHI-aims \cite{FHI-AIM-URL, FHI-AIM-PAPER}, 
which is a well-known electronic state calculation code. The
EigenExa code, on the other hand, was developed at RIKEN in Japan. Importantly, the ELPA code has routines optimized for x86, IBM 
Blue-Gene, and AMD architectures, 
whereas the EigenExa code was developed to be optimal mainly on the K computer, which is
a Japanese flagship supercomputer.
Both ScaLAPACK and ELPA provide the reducer routines, 
and all of ScaLAPACK, ELPA, and EigenExa provide the SEP solver routines.

Since the computational performance depends both on the problem and the architecture, 
it is, in principle, possible to construct
a `hybrid' workflow in which the reducer routine is chosen from one library and the SEP solver routine is chosen from another library,
so as to realize optimal performance.
The middleware EigenKernel was developed in order to realize such hybrid workflows. 
An obstacle to realizing the hybrid workflow is the difference of matrix distribution schemes between different libraries. 
EigenKernel provides data conversion routines between libraries and surmounts this obstacle. 

Figure \ref{FIG-WORKFLOW-EK} shows the possible workflows for a future version of EigenKernel with the a posteriori verification routine.
The SEP solvers and the reducers in Fig. \ref{FIG-WORKFLOW-EK} are briefly explained. 
The SEP solver and the two reducers in ScaLAPACK are the traditional routines.
The SEP solvers of  `ELPA1' and `Eigen\verb|_s|' are also based 
on the traditional algorithm with tridiagonalization. 
The other two solvers `ELPA2' and `Eigen\verb|_sx|' 
and the reducer in ELPA are based on 
non-traditional algorithms for better performance in massive parallelism.
The detailed algorithms for these routines are found in Ref.~\cite{2018TANAKA}.

\begin{figure}[htb]
\begin{center}
  \includegraphics[width=0.5\textwidth]{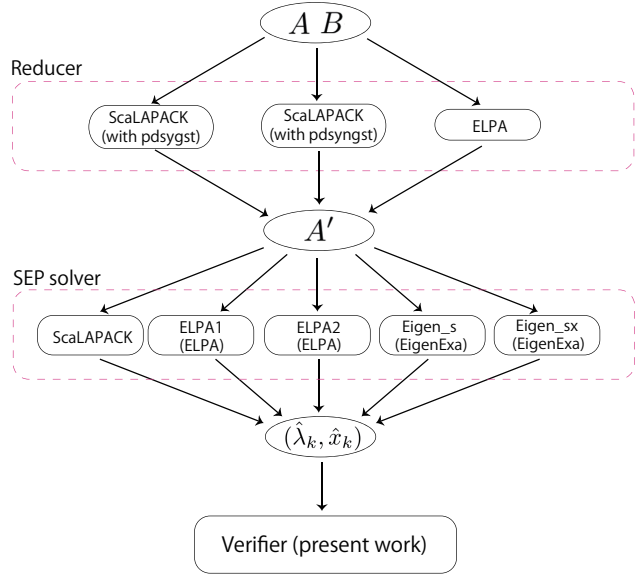}
\end{center}
\caption{Schematic diagram of the possible hybrid workflows for a future version of EigenKernel with the a posteriori verification routine.
Two routines in ScaLAPACK and one routine in ELPA are available for the reducer,
whereas one routine in ScaLAPACK, two routines in ELPA, and two routines in EigenExa are available for the SEP solver. 
The a posteriori verification routine is commonly used among the workflows. 
}
\label{FIG-WORKFLOW-EK}       
\end{figure}

\subsection{Verified numerical computations}
\label{ssec:VNC}
We briefly explain how to obtain mathematically rigorous numerical results using floating-point arithmetic.
Let $\mathbb{F}$ and $\mathbb{IF}$ be sets of floating-point numbers and intervals, respectively.
We use bold-faced letters for interval matrices, the elements of which are intervals.
For an interval matrix $\mathbf{C}$, $C_{\mathrm{inf}}$ and $C_{\mathrm{sup}}$ denote the left and right endpoints, respectively, such that ${\bf C}=[C_{\mathrm{inf}},C_{\mathrm{sup}}]$, i.e., $\mathbf{C}_{ij} = [(C_{\mathrm{inf}})_{ij},(C_{\mathrm{sup}})_{ij}]$ for all $(i,j)$ pairs, which is known as ``inf-sup'' form.
In addition, $C_{\mathrm{mid}}$ and  $C_{\mathrm{rad}}$ denote the midpoint and the radius of ${\bf C}$, respectively, such that $\mathbf{C} = [C_{\mathrm{mid}} - C_{\mathrm{rad}},C_{\mathrm{mid}} + C_{\mathrm{rad}}]$, which is known as ``mid-rad'' form.
Let $\mathit{fl}(\cdot)$, $\mathit{fl}_\bigtriangledown(\cdot) $, and $\mathit{fl}_\bigtriangleup(\cdot) $ be computed results by floating-point arithmetic as defined in IEEE 754 with rounding to the nearest (roundTiesToEven), rounding downwards (roundTowardNegative), and rounding upwards (roundTowardPositive), respectively.
For a given matrix $C = (c_{ij}) \in \Rnn$, the notation $\abs{C}$ indicates $\abs{C} = (\abs{c_{ij}}) \in \Rnn$, and the same applies to vectors, i.e., the absolute value is taken componentwise.

Next, we review basic interval matrix multiplication (cf.~\cite{Ru1999}).
For two point matrices $P, Q \in \mathbb{F}^{n \times n}$, 
the matrix product $PQ \in \mathbb{R}^{n \times n}$ can be enclosed as
\begin{equation}
PQ \in [\mathit{fl}_\bigtriangledown (PQ), \ \mathit{fl}_\bigtriangleup (PQ)],
\label{eq:pp}
\end{equation}
where two matrix multiplications are required.
For a point matrix $P \in \mathbb{F}^{n \times n}$ and an interval matrix ${\bf Q} \in \mathbb{IF}^{n \times n}$, 
the product $P{\bf Q}$ can efficiently be enclosed using mid-rad form of ${\bf Q}$ as
\begin{equation}
P{\bf Q} \subset [\mathit{fl}_\bigtriangledown(PQ_{\mathrm{mid}} - T), \ \mathit{fl}_\bigtriangleup (PQ_{\mathrm{mid}} + T)], \quad T = \mathit{fl}_\bigtriangleup (|P|Q_{\mathrm{rad}}),
\label{eq:pi}
\end{equation}
which involves three matrix multiplications.
Although the inf-sup form can also be used for calculating the enclosure of $P{\bf Q}$, the inf-sup form cannot be written with products of point matrices simply, so that it is much more difficult for the inf-sup form to achieve high-performance in practice, as compared to the mid-rad form~\cite{Ru1999}.
If $\bf Q$ is given by the inf-sup form $[Q_{\mathrm{inf}},Q_{\mathrm{sup}}]$, we can easily transform $\bf Q$ into the mid-rad form, for example, by
\[
  Q_{\mathrm{mid}} = \mathit{fl}_\bigtriangleup((Q_{\mathrm{inf}} + Q_{\mathrm{sup}})/2), \quad Q_{\mathrm{rad}} = \mathit{fl}_\bigtriangleup(Q_{\mathrm{mid}} - Q_{\mathrm{inf}}),
\]
which satisfies $[Q_{\mathrm{inf}},Q_{\mathrm{sup}}] \subset [Q_{\mathrm{mid}}-Q_{\mathrm{rad}},Q_{\mathrm{mid}}+Q_{\mathrm{rad}}]$.

There exist several implementations of the above interval arithmetic for matrix multiplication, e.g., C-XSC~\cite{C-XSC}, a C++ library, and INTLAB~\cite{INTLAB}, a Matlab/Octave toolbox for verified numerical computations.
Both C-XSC and INTLAB share the common feature that they use Basic Linear Algebra Subprograms (BLAS) routines.
In other words, we can efficiently implement interval matrix multiplication using PBLAS, the parallel version of BLAS, on distributed computing environments, as long as directed rounding in floating-point operations is available in BLAS routines for matrix multiplication and the reduction operation of summation.

\section{A posteriori verification methods \label{SEC-VERIF-METHODS}}

\subsection{Possible verification methods \label{SEC-VERIF-METHODS-GENERAL}}

Possible verification methods are discussed here. 
In order to measure the accuracy of the computed solution $(\ap{\lambda}_{k},\ap{x}_k)$, application researchers often compute a norm of the residual vector, such as
\[
\frac{\normtwo{A \ap{x}_k - \ap{\lambda}_k B \ap{x}_k}}{\normtwo{\ap{x}_k}} .
\]
Although this quantity usually suffices to check whether the solver works correctly, it does not verify the accuracy of the computed eigenvalue.
The following inequality is a known residual bound \cite{MiOgRuOi2010}:
\begin{eqnarray}
\min_{1 \le j \le n}\abs{\lambda_{j} - \ap{\lambda}_{k}} \le 
\sqrt{\normtwo{B^{-1}}}\frac{\normtwo{A \ap{x}_k - \ap{\lambda}_k B \ap{x}_k}}{\sqrt{\trans{\ap{x}_k}B\ap{x}_k}},
 \label{EQ-GEP-BOUND}
\end{eqnarray}
which is straightforwardly derived from Wilkinson's bound~\cite{Wi1961} for the standard eigenvalue problem.
From the bound \eqref{EQ-GEP-BOUND}, we can confirm that some eigenvalue of $(A,B)$ exists in the neighborhood of $\ap{\lambda_k}$ satisfying (\ref{EQ-GEP-BOUND}).
However, we cannot determine whether $\ap{\lambda}_{k}$ is an approximation of the $k$-th eigenvalue of $(A,B)$.
In order to understand the electronic state of the problems correctly, it is crucial to determine the order of eigenvalues \cite{LeHoSoMiZh2018}.

To our knowledge, we have the following two approaches to determine the order of eigenvalues of symmetric matrices:
\begin{itemize}
    \item[(a)] Compute all eigenpairs (pairs of eigenvalues and eigenvectors) and verify the error bounds of all computed eigenvalues (cf.,\ e.g.\, \cite{MiOgRuOi2010,Mi2012}).
    \item[(b)] Compute an approximation $\ap{\lambda}_{k}$ of a target eigenvalue using Sylvester's law of inertia with $\mathrm{LD\trans{L}}$ decomposition \cite{LeHoSoMiZh2018}, and verify that $\ap{\lambda}_{k}$ is an approximation of the $k$-th eigenvalue with an error bound (cf.,\ e.g.\, \cite{Ya2001}).
\end{itemize}
The advantages and disadvantages of each approach from a practical point of view are as follows:
\begin{itemize}
\item Approach (a) is simpler, numerically more stable, and easier to implement than Approach (b).
\item Approach (a) can straightforwardly use highly optimized routines for matrix multiplication and eigenvalue decomposition.
\item Approach (a) cannot exploit the sparsity of $A$ and $B$, whereas Approach (b) can to a certain extent.
\end{itemize}
In the present paper, we adopt Approach (a) from the aspect of simplicity and efficiency of code development on supercomputers.

\subsection{Proposed method}
We attempt to obtain componentwise error bounds for computed eigenvalues $\ap{\lambda}_{k}$, $k = 1, 2, \dots, n$.
Let $X, D \in \Rnn$ denote a matrix comprising all generalized eigenvectors of $(A,B)$ and a diagonal matrix of the corresponding generalized eigenvalues such that
\[
X = [x_{1},x_{2},\dots,x_{n}], \quad D = \diag(\lambda_{1},\lambda_{2},\dots,\lambda_{n}) .
\]
Let $I$ denote the $n \times n$ identity matrix.
Then, we have
\[
 \left\{\begin{array}{l}
 AX = BXD, \\
 \trans{X}BX = I .
 \end{array}\right.
\]
Let $\ap{X} = [\ap{x}_{1},\ap{x}_{2},\dots,\ap{x}_{n}] \in \Rnn$ and $\ap{D} = \diag(\ap{\lambda}_{1},\ap{\lambda}_{2},\dots,\ap{\lambda}_{n}) \in \Rnn$ be approximations of $X$ and $D$, respectively.
Suppose $\ap{X}$ is nonsingular.
Then,
\[
  A\ap{X} \approx B\ap{X}\ap{D}, \quad \trans{\ap{X}}B\ap{X} \approx I \quad \Rightarrow \quad \ap{X}^{-1}B^{-1}A\ap{X} \approx \ap{D}, \quad (B\ap{X})^{-1} \approx \trans{\ap{X}} .
\]
Since $\ap{X}^{-1}B^{-1}A\ap{X}$ is a similarity transformation of $B^{-1}A$, the eigenvalues of $\ap{X}^{-1}B^{-1}A\ap{X}$ are the same as those of $B^{-1}A$, and thus the generalized eigenvalues of $(A,B)$.

Here, we attempt to compute an inclusion of $\ap{X}^{-1}B^{-1}A\ap{X}$.
To this end, we introduce Yamamoto's theorem for verified solutions of linear systems.
For given matrices $P = (p_{ij}), Q = (q_{ij}) \in \Rnn$, the notation $P \le Q$ indicates $p_{ij} \le q_{ij}$ for all $(i,j)$, and the same applies to vectors, i.e., the inequality holds componentwise.
Moreover, define $e \equiv \trans{(1,1,\dots,1)} \in \Rn$.

\begin{theorem}[Yamamoto \cite{Ya1984}]
\label{th:yamamoto}
Let $A$ and $C$ be real $n \times n$ matrices, and let $b$ and $\ap{x}$ be real $n$-vectors.
If $\norminf{I - CA} < 1$, then $A$ is nonsingular, and
\[
  \abs{A^{-1}b - \ap{x}} \le \abs{C(b - A\ap{x})} + \frac{\norminf{C(b - A\ap{x})}}{1 - \norminf{I - CA}}\abs{I - CA}e .
\]
\end{theorem}
\noindent
In practice, we adopt an approximate inverse of $A$ as $C$ in Theorem~\ref{th:yamamoto}.

In order to apply Yamamoto's theorem to componentwise error bounds for computed eigenvalues with the Gershgorin circle theorem, we present a variant of Yamamoto's theorem.
\begin{theorem}
\label{th:proposed}
Let $A$, $B$, $C$, and $\ap{X}$ be real $n \times n$ matrices.
If $\norminf{I - CA} < 1$, then $A$ is nonsingular, and
\[
  \abs{A^{-1}B - \ap{X}}e \le \abs{C(B - A\ap{X})}e + \frac{\norminf{C(B - A\ap{X})}}{1 - \norminf{I - CA}}\abs{I - CA}e .
\]
\end{theorem}
\begin{proof}
In a similar manner to the derivation of Yamamoto's theorem, and noting that for $P \in \Rnn$, $\abs{P}e \le \norminf{P}e$, we have
\begin{eqnarray*}
\abs{A^{-1}B - \ap{X}}e &=& \abs{(CA)^{-1}C(B - A\ap{X})}e \\
 &\le& \abs{(CA)^{-1}}\cdot\abs{CR}e, \quad R \equiv B - A\ap{X} \\
 &=& \abs{(I - (I - CA))^{-1}}\cdot\abs{CR}e
 = \abs{I + G + G^{2} + \cdots}\cdot\abs{CR}e, \quad G \equiv I - CA \\
 &\le& \abs{CR}e + \abs{G}(I + \abs{G} + \abs{G}^{2} + \cdots)\abs{CR}e \\
 &\le& \abs{CR}e + \norminf{CR}\abs{G}(I + \abs{G} + \abs{G}^{2} + \cdots)e \\
&\le& \abs{CR}e + \frac{\norminf{CR}}{1 - \norminf{G}}\abs{G}e,
\end{eqnarray*}
which proves the theorem.
\end{proof}

We now consider a linear system $(B\ap{X})Y = A\ap{X}$ for $Y$.
Then, we can regard $\ap{D}$ as its approximate solution and $\trans{\ap{X}}$ as an approximate inverse of $B\ap{X}$.
Let $Y$, $R$, and $G$ be defined as
\begin{equation}
    \label{eq:RGdef}
 \left\{\begin{array}{l}
 R \equiv \trans{\ap{X}}(A\ap{X} - B\ap{X}\ap{D}), \\
 G \equiv \trans{\ap{X}}B\ap{X} - I .
 \end{array}\right.
\end{equation}
If $\norminf{G} < 1$, applying Theorem~\ref{th:proposed} to the linear system $(B\ap{X})Y = A\ap{X}$ yields
\begin{equation}
\label{eq:estimate}
  \abs{\ap{X}^{-1}(B^{-1}A)\ap{X} - \ap{D}}e = \abs{(B\ap{X})^{-1}(A\ap{X}) - \ap{D}}e \le \abs{R}e + \frac{\norminf{R}}{1 - \norminf{G}}\abs{G}e \equiv r .
\end{equation}
Recall that $\lambda_{i}$, $i = 1, 2, \dots, n$, are the eigenvalues of $B^{-1}A$.
For $\Lambda \equiv \{\lambda_{1},\lambda_{2},\dots,\lambda_{n}\}$, the Gershgorin circle theorem implies
\begin{equation}
\label{eq:veig}
  \Lambda \subseteq \bigcup_{k = 1}^{n}[\ap{\lambda}_{i} - r_{i}, \ap{\lambda}_{i} + r_{i}] .
\end{equation}
If all the disks $[\ap{\lambda}_{i} - r_{i}, \ap{\lambda}_{i} + r_{i}]$ are isolated, then all of the eigenvalues are separated, i.e., each disk contains precisely one eigenvalue of $B^{-1}A$~\cite[pp.~71ff]{Wi1965},
as shown schematically in Fig.~\ref{FIG-VERIFIED-SOLUTION}.
If several disks are overlapped such that $|\ap{\lambda}_{k+1} - \ap{\lambda}_{k}| > r_k + r_{k+1}$ for some $k$, 
then some of the eigenvalues are degenerate or nearly degenerate.
Moreover, if $B$ is ill-conditioned, then the $B$-orthogonality of $\ap{X}$ may break down such that $\norminf{G} \ge 1$.
In such a case, Theorem~\ref{th:proposed} cannot be applied, and the verification procedure must end in failure.
Therefore, we need to check whether $\norminf{G} < 1$ in code development from the verification method.

In \cite{Mi2012}, a similar method has been proposed, which is essentially the same as the proposed method.
The main difference between the method in \cite{Mi2012} and the proposed method is that the former focuses on the non-symmetric case and is more general.
On the other hand, the proposed method is specialized for the symmetric case, i.e., we can avoid complex arithmetic including the verification procedure and compute an approximate inverse of $B\ap{X}$ by utilizing $\trans{\ap{X}} \approx (B\ap{X})^{-1}$.

\begin{figure}[t]
\begin{center}
  \includegraphics[width=0.35\textwidth]{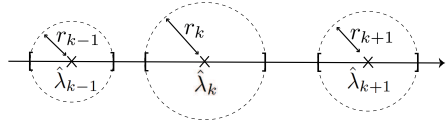}
\end{center}
\caption{Schematic diagram of verified solution when all of the disks are separated, and each disk contains precisely one eigenvalue ($\lambda_k \in [\ap{\lambda}_k- r_k, \ap{\lambda}_k + r_k ]$).
}
\label{FIG-VERIFIED-SOLUTION}       
\end{figure}

\subsection{Code development \label{SEC-VERIF-CODE}}

We explain how to obtain an upper bound of the vector $r$ in \eqref{eq:estimate} using only floating-point arithmetic.
We first attempt to obtain upper bounds $G'$ and $R'$ of $| R | = | \hat{X}^{\top}(A\hat{X}- B\hat{X} \hat{D})|$
and $| G | = | \hat{X}^{\top} B \hat{X}-I |$ in (\ref{eq:RGdef}) such that $\abs{G} \le G'$ and $\abs{R} \le R'$ as follows:
\begin{enumerate}
\item ${\bf C} \leftarrow B \hat{X}$ \% Two matrix multiplications based on (\ref{eq:pp})
\item ${\bf F} \leftarrow \hat{X}^{\top}{\bf C}$ \% Three matrix multiplications based on (\ref{eq:pi})
\item ${\bf W} \leftarrow {\bf F} - I$ \% Negligible cost, $W_{\mathrm{inf}}\equiv\mathit{fl}_\bigtriangledown (F_{\mathrm{inf}}-I), \ W_{\mathrm{sup}}\equiv\mathit{fl}_\bigtriangleup  (F_{\mathrm{sup}}-I)$
\item $| G | \le \max(|W_{\mathrm{inf}}|,|W_{\mathrm{sup}}|) \equiv G'$
\item ${\bf F} \leftarrow A \hat{X}$ \% Two matrix multiplications based on (\ref{eq:pp})
\item ${\bf C} \leftarrow {\bf C} \hat{D} $ \% Negligible cost because $\hat D$ is a diagonal matrix
\item ${\bf C} \leftarrow {\bf F} - {\bf C} $ \% Negligible cost, \ $C_{\mathrm{inf}}$ is overwritten by $\mathit{fl}_\bigtriangledown (F_{\mathrm{inf}}-C_{\mathrm{sup}}), \ C_{\mathrm{sup}}$ is overwritten by $\mathit{fl}_\bigtriangleup  (F_{\mathrm{sup}}-C_{\mathrm{inf}})$
\item ${\bf C} \leftarrow \hat{X}^{\top} {\bf C} $ \% Three matrix multiplications based on (\ref{eq:pi})
\item $| R | \le \max(|C_{\mathrm{inf}}|,|C_{\mathrm{sup}}|) \equiv R'$
\end{enumerate}
Note that the notation `$\leftarrow$' indicates enclosure of the result.
Moreover, for given matrices $P =(p_{ij}), Q = (q_{ij}) \in \mathbb{F}^{n \times n}$, the notation $\max(P,Q)$ indicates $\max(p_{ij},q_{ij})$ for all $(i,j)$ pairs, i.e., the maximum is taken componentwise.
Here, five matrix multiplications are required for calculating $G'$ until Step 4, and an additional five matrix multiplications for the remaining calculations.
Thus, in total, 10 matrix multiplications are required for calculating $G'$ and $R'$.
Therefore, calculating $G'$ and $R'$ involves $20n^3 + \mathcal{O}(n^2)$ floating-point operations if the symmetry of $G$ is not taken into account.
We compute the upper bounds of $\| R \|_\infty$ and $\| G \|_\infty$ as
\[
\| R \|_\infty \le \| R' \|_\infty \le \mathit{fl}_\bigtriangleup (\| R' \|_\infty) \equiv  \alpha_1, \quad 
\| G \|_\infty \le \| G' \|_\infty \le \mathit{fl}_\bigtriangleup (\| G' \|_\infty) \equiv  \alpha_2.
\]
If $\alpha_2 \ge 1$, then the verification failed.
Hence, we check $\alpha_2 < 1$ or $\alpha_2 \ge 1$ after Step 4.
If $\alpha_2 \ge 1$, then the computation prematurely finishes without proceeding to Step 5.
Otherwise, we proceed until Step 9 and obtain upper bound $r'$ of $r$ in \eqref{eq:estimate} by
\begin{equation}
r \le \mathit{fl}_\bigtriangleup \left( R'e + \frac{\alpha_1}{\mathit{fl}_\bigtriangledown (1-\alpha_2)}G'e \right) \equiv r'.
\label{eq:radius}
\end{equation}

The routine \textsf{pdsygvx} in ScaLAPACK produces computed eigenvalues $\ap{\lambda}_{i}$ with $\ap{\lambda}_1 \le \ap{\lambda}_2 \le \dots \le \ap{\lambda}_n$.
Therefore, if $\ap{\lambda}_{i+1} - \ap{\lambda}_{i} > r'_{i} + r'_{i+1}$ are satisfied for all $i = 1, 2, \dots, n-1$, 
then we can separate all of the eigenvalues and determine the order of the eigenvalues correctly.

The test code was developed in the C language
with the parallel libraries PBLAS and ScaLAPACK.
The solver procedure uses
a GEP solver routine (\textsf{pdsygvx}) in ScaLAPACK, whereas
the verifier routine uses 
the matrix multiplication routine (\textsf{pdgemm}) in PBLAS.

Note that 
the verifier procedure is based primarily on
matrix multiplication,
whereas the solver procedure consists
of complicated procedures, such as 
Cholesky decomposition, and tridiagonalization. 
Therefore, 
the verifier procedure is expected to be
moderate in terms of computational time 
and to be efficient in terms of parallelism,
as compared to the solver procedure.

\section{Numerical example \label{SEC-NUM-EXAMPLE} }
\subsection{Problem \label{SEC-NUM-EXAMPLE-PROBLEM} }

Numerical examples are presented in this section. 
All matrix eigenvalue problems stem from the electronic-state calculation software ELSES \cite{ELSES-URL, HOSHI2012-ELSES, HOSHI2016-SC16},
and the matrix data files appear in the ELSES matrix library
\cite{ELSES-MATRIX-LIBRARY-URL, HOSHI2018-PENTA}.
Details are explained in \ref{SEC-GHEV}.
The problems calculated in this section are 
PPE354, PPE3594, PPE7194, PPE17994, PPE107994, VCNT22500, VCNT225000, and NCCS430080
in the ELSES matrix library. 
The matrices are those of systems having disordered atomic structures.  
Disordered systems are important for industrial applications because most industrial materials are disordered, unlike ideal crystal or periodic structures.  
Consequently, eigenvalues are not degenerate in all of the problems. 
The number in the problem name indicates the matrix dimension $n$.
For example, the system PPE354 contains $n \times n$ matrices $A$ and $B$ 
with $n = 354$. 
All of the matrices $A$ and $B$ in these systems are real symmetric. 
The systems with the letters `PPE' are systems
of organic polymers of poly-(phenylene-ethynylene) (PPE).
The left-hand panel of Fig.~\ref{FIG-EigenValue-Graph}(a) shows the structural formula of PPE, and the right-hand panel of Fig.~\ref{FIG-EigenValue-Graph}(b) shows a part of the polymer in a disordered structure. 
The difference of the matrix size stems from the length of the polymer chain.
The system of PPE354 is, for example, 
a polymer with $N_{m}=10$ monomers and $N_{\rm atom}= 12N_{m} = 120$ atoms. 
The system VCNT225000 is the system of vibrating carbon nanotube (VCNT). 
The system NCCS430080 is the system of nano-composite carbon solid (NCCS)
\cite{HOSHI2013-JPSJ} and will be explained in the last paragraph of this section. 

The characteristic of the eigenvalue distribution 
can be captured by the following two quantities.
One is the difference of sequential approximate eigenvalues 
$\ap{\delta}_k \equiv \ap{\lambda}_{k+1} - \ap{\lambda}_k$, $k=1,2,...,n-1$, and the other is 
the eigenvalue count $I(\lambda)$, which is defined on the eigenvalue axis $\lambda$ as 
\begin{eqnarray}
I(\lambda) \equiv \sum_{k = 1}^{n} \theta(\lambda_k - \lambda )
\end{eqnarray}
with the step function 
\begin{eqnarray}
\theta(\lambda) \equiv 
\begin{cases}
1 & ( \lambda \ge 0) \\
0 & ( \lambda < 0).
\end{cases}
\end{eqnarray}
In other words,  
the eigenvalue count $I(\lambda)$ is 
the number of the eigenvalues that are smaller than $\lambda$.

Here, we demonstrate the similarity and the size dependence of the eigenvalue distribution 
among the organic polymer systems.
The organic polymers of PPE354, PPE17994, and PPE107994 are selected.  
Figures \ref{FIG-EigenValue-Graph}(b) and  \ref{FIG-EigenValue-Graph}(c) show 
the normalized eigenvalue distribution $I(\lambda)/n$ among these three systems.
The three polymers exhibit quite similar curves in Figs.~\ref{FIG-EigenValue-Graph}(b) and ~\ref{FIG-EigenValue-Graph}(c), and, therefore, 
the difference $\ap{\delta}_k$ 
is nearly proportional to $1/n$ $(\ap{\delta}_k \propto 1/n) $,
as explained in Section~\ref{SEC-INTRO}.

\begin{figure}[thb]
\begin{center}
  \includegraphics[width=0.7\textwidth]{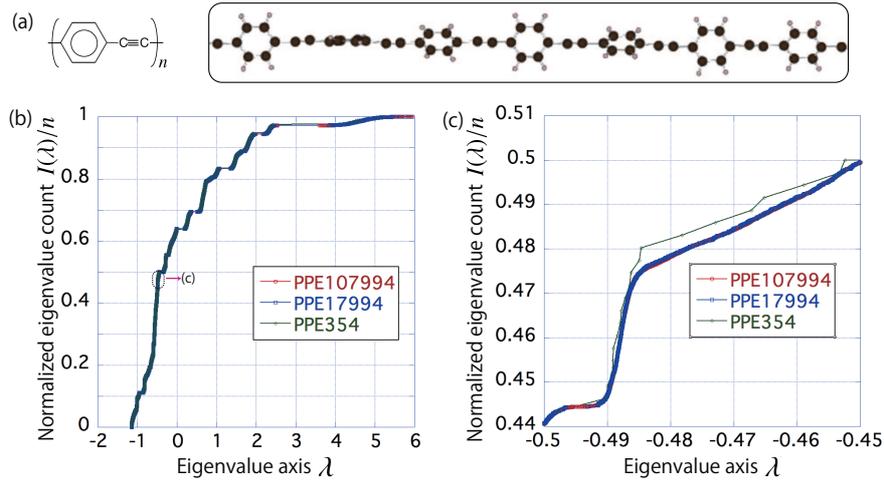}
\end{center}
\caption{(a) Structural formula (left) and a part of the atomic structure (right) of poly-(phenylene-ethynylene) (PPE). 
(b) Similarity of eigenvalue distribution in PPE354 (circle), PPE17994 (square), and PPE107994 (diamond). 
The normalized eigenvalue counts $I(\lambda)/n$ are plotted on the eigenvalue axis $\lambda$.
(c) Close-up of the local area indicated by the dotted line in (b). 
}
\label{FIG-EigenValue-Graph}       
\end{figure}

\subsection{Numerical results \label{SEC-NUM-EXAMPLE-RESULT} }

Tables \ref{TABLE-NUM-EXAMPLE-RESULT} and \ref{TABLE-NUM-EXAMPLE-TIME} show the calculation results
on the K computer.
First, we focus on the numerical results for the approximate eigenvalues $\ap{\lambda}$ and its upper bound $r'$.
The routine \textsf{pdsygvx} in ScaLAPACK produces $\ap{\lambda}_i$, $i = 1, 2, \dots, n$ with $\ap{\lambda}_1 \le \ap{\lambda}_2 \le \dots \le \ap{\lambda}_n$.
The vector $r'$ is obtained by (\ref{eq:radius}).
Here, we define 
the radius sum 
$\rho_k \equiv r'_{k+1} + r'_{k}$ for $k = 1, 2, \dots, n-1$.
We find $m$ such that $\displaystyle \ap{\delta}_m - \rho_m = \min_{1 \le k \le n - 1}(\ap{\delta}_k - \rho_k)$.
The items ``Difference'' and ``Radius'' in Table~\ref{TABLE-NUM-EXAMPLE-RESULT} show
$\ap{\delta}_m$ and $\rho_m$, respectively.
As shown in the table, $\ap{\delta}_m > \rho_m$ is satisfied in all of the problems,
or all of the disks of $|\lambda_k - \ap{\lambda}_k| < r'_k$ are separated 
as in Fig.~\ref{FIG-VERIFIED-SOLUTION}.
Thus, we can determine the order of eigenvalues in each problem.
If $\ap{\delta}_k < \rho_k$ is satisfied for some $k$, 
then the two disks of 
$|\lambda_k - \ap{\lambda}_k| < r'_k$ and 
$|\lambda_{k+1} - \ap{\lambda}_{k+1}| < r'_{k+1}$
are overlapped and 
the two exact eigenvalues of $\lambda_k$ and $\lambda_{k+1}$ may degenerate.

Figure~\ref{FIG-EigenValue-Graph2}(a) 
shows 
the eigenvalue difference $\{ \ap{\delta}_k \}$ and the radius sum $\{ \rho_k \}$ as a function of the eigenvalue $\{ \ap{\lambda}_k \}$
in the case of PPE107994. 
The radius sum satisfies $\rho_k \le 10^{-10}$
and is smaller than the difference ($\rho_k < \ap{\delta}_k$). 
We found the minimality $m=49,201$ and 
$\ap{\lambda}_{49201} \approx -0.488$,
$\ap{\delta}_{49201} \approx 6.42 \times 10^{-11}$, and
$\rho_{49201} \approx 9.17 \times 10^{-12}$.
Figure~\ref{FIG-EigenValue-Graph2}(b)
shows a close-up of 
Fig.~\ref{FIG-EigenValue-Graph2}(a) 
and contains the eigenvalue 
$\ap{\lambda}_{49201} \approx -0.488$.
It is reasonable that
the eigenvalue $\ap{\lambda}_{49201}$ appears
in the region of $-0.490 < \lambda < -0.485 $, 
because many eigenvalues are densely clustered, and 
the eigenvalue count $I(\lambda)$ increases rapidly in the region, as shown in 
Fig.~\ref{FIG-EigenValue-Graph}(c).
The same analysis was also carried out 
in the case of NCCS430080, which is the largest problem among the present calculations,  
and the results are shown in Figs.~\ref{FIG-EigenValue-Graph2}(c) and \ref{FIG-EigenValue-Graph2}(d). 
The radius sum is smaller than the difference ($\rho_k < \ap{\delta}_k$).

\begin{table}[htb]
  \begin{center}
  \caption{Numerical example  \label{TABLE-NUM-EXAMPLE-RESULT}}
  \begin{tabular}{|l|r|l|l|} \hline
    Problem name & Matrix dimension ($n$) & Difference ($\ap{\delta}_m$) & Radius sum ($\rho_m$)  \\ \hline \hline
    PPE354 & 354 & $6.61 \times 10^{-5}$ & $4.90 \times 10^{-13}$  \\
    PPE3594 & 3,594 & $1.03 \times 10^{-7}$ & $1.33 \times 10^{-12}$  \\
    PPE7194 & 7,194 & $5.55 \times 10^{-8}$ & $1.18 \times 10^{-12}$  \\
    PPE17994 & 17,994 &$5.32 \times 10^{-11}$ & $2.56 \times 10^{-12}$  \\
    PPE107994 & 107,994 & $6.42 \times 10^{-11}$ & $9.17 \times 10^{-12}$  \\
    VCNT22500 & 22,500 & $2.59 \times 10^{-7}$ & $3.20 \times 10^{-10}$  \\
    VCNT225000 & 225,000 & $1.97 \times 10^{-9}$ & $1.64 \times 10^{-9}$  \\
    NCCS430080 & 430,080 & $5.10 \times 10^{-9}$ & $1.61 \times 10^{-9}$  \\
    \hline
  \end{tabular}
  \end{center}
\end{table}

\begin{figure}[tb]
\begin{center}
  \includegraphics[width=0.7\textwidth]{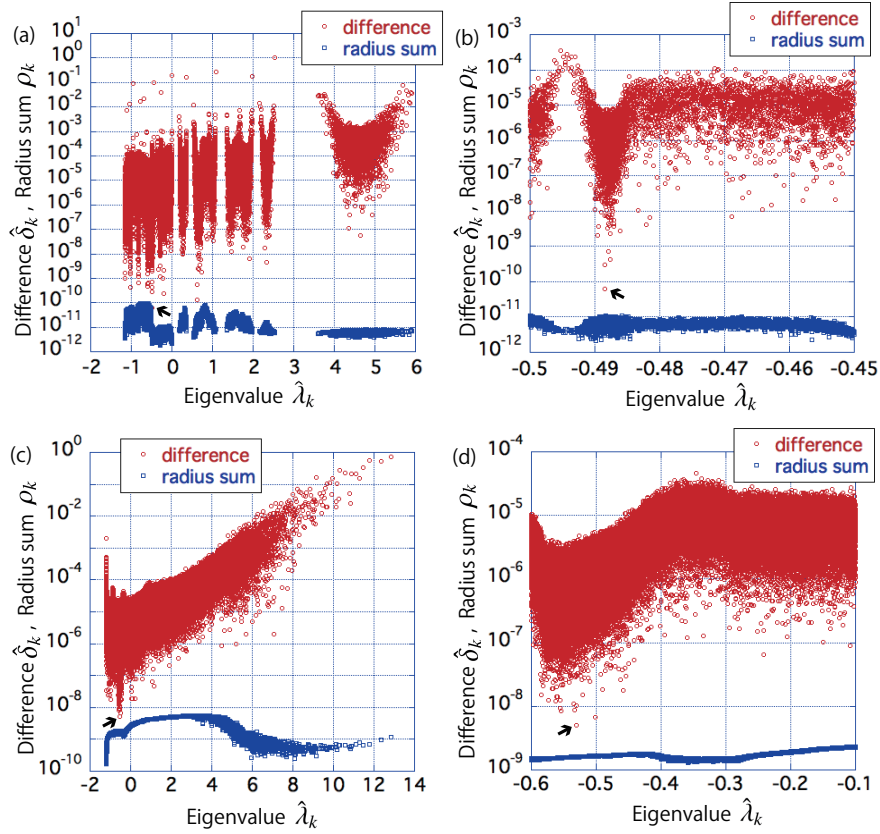}
\end{center}
\caption{(a) Plot of the eigenvalue difference $\{ \ap{\delta}_k \}$ and the radius sum $\{ \rho_k \}$ as a function of the eigenvalue $\{ \ap{\lambda}_k \}$ in the case of PPE107994. The arrow indicates $\ap{\delta}_m$.
(b) Close-up of (a). The arrow indicates $\ap{\delta}_m$.
(c) Plot of the eigenvalue difference $\{ \ap{\delta}_k \}$ and the radius sum $\{ \rho_k \}$ as a function of the eigenvalue $\{ \ap{\lambda}_k \}$ in the case of NCCS430080. The arrow indicates $\ap{\delta}_m$.
(d) Close-up of (c). The arrow indicates $\ap{\delta}_m$.
}
\label{FIG-EigenValue-Graph2}       
\end{figure}
Table~\ref{TABLE-NUM-EXAMPLE-TIME} shows the computational times.
The item $T_{\rm sol}$ in Table~\ref{TABLE-NUM-EXAMPLE-TIME} shows the computing time for \textsf{pdsygvx} in ScaLAPACK.
The item $T_{\rm veri}$ shows the computing time for the verification process, mainly, the time for matrix multiplications.
Here, the verifier consumes a moderate cost ($T_{\rm veri} \le T_{\rm sol}$),
as expected in Section~\ref{SEC-VERIF-CODE}.
More intensive benchmarks, including weak scaling, will be carried out in the future. 

\begin{table}[bht]
  \begin{center}
  \caption{Elapsed times among the problems. The number of used processor nodes $P$ and the elapsed times for solver $T_{\rm sol}$ and
   verifier $T_{\rm veri}$ are shown.  \label{TABLE-NUM-EXAMPLE-TIME}}
  \begin{tabular}{|l|r|r|r|} \hline
    Problem name & $P$ & $T_{\rm sol}$ & $T_{\rm veri}$ \\ \hline \hline
    PPE354 & 4 & 0.32 & 0.12  \\
    PPE3594 & 4 & 20.74 & 4.73  \\
    PPE7194 & 4 & 118.84 & 31.74  \\
    PPE17994 & 16 & 217.91 & 105.75  \\
    PPE107994 & 600 & 1009.85 & 682.92  \\
    VCNT22500 & 64 & 105.75 & 59.06  \\
    VCNT225000 & 2025 & 2625.76 & 1775.09  \\
    NCCS430080 & 6400 & 8960.03 & 3496.56   \\
    \hline
  \end{tabular}
  \end{center}
\end{table}

In conclusion,
the verification procedure delivers
the intervals that contain the exact eigenvalues
($|\lambda_k - \ap{\lambda}_k| < r'_k$) 
with the approximate eigenvalues $\ap{\lambda}_k$
and the radius $r'_k$. 
We plan to upload 
the radius data files
in ELSES matrix library, 
as well as the input matrix data
and the approximate eigenvalue data.
Then, a graph
similar to Fig.~\ref{FIG-EigenValue-Graph2} can be drawn in order to measure the accuracy of the computed solutions.

Finally,  the present numerical results are discussed  in the context of computational physics.
The matrix problem of NCCS430080, the largest matrix problem in the present paper, 
appears in a previous paper
on a nano-composite carbon solid \cite{HOSHI2013-JPSJ}. 
In general, carbon can form diamond and graphite crystals. 
The material is composed of graphite-like and diamond-like domains. 
Figure~\ref{FIG-NPD-WFN} shows an example of the electronic wavefunction (the highest occupied electronic wavefunction or the wavefunction of the electron that has the highest energy). 
The atomic structure of Fig.~\ref{FIG-NPD-WFN} is that of Fig. 2(a) of Ref.~\cite{HOSHI2013-JPSJ}. (See Ref.~\cite{HOSHI2013-JPSJ} for details.) 
In the present context,
Fig.~\ref{FIG-NPD-WFN} indicates that the wavefunction
is an intermediate wavefunction, as explained in Section~\ref{SEC-CLUSTERED}, and lies in the boundary region between graphite-like and diamond-like domains.
The a posteriori verification procedure confirms that 
all of the eigenvalues are distinguished numerically, 
and the above physical discussion regarding each wavefunction is meaningful.

\begin{figure}[ht]
\begin{center}
  \includegraphics[width=0.5\textwidth]{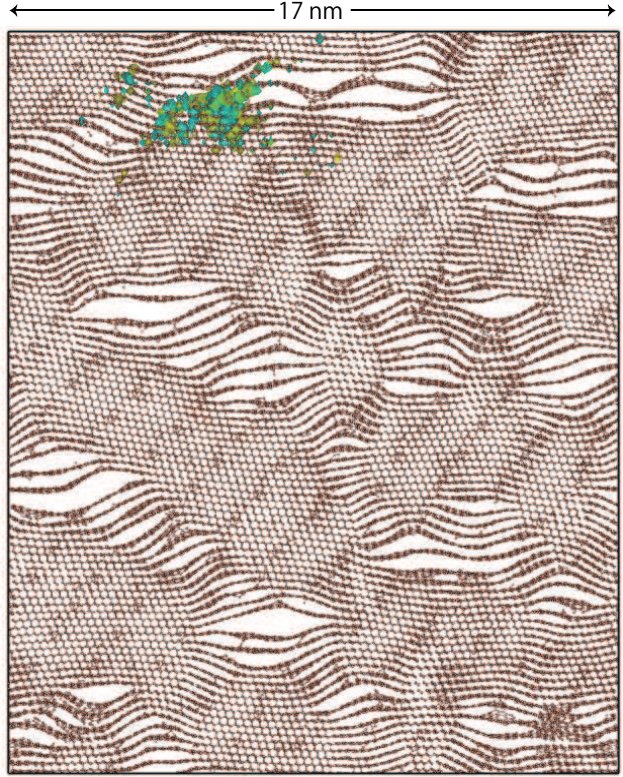}
\end{center}
\caption{Example of an electronic wavefunction of a nano-composite carbon solid \cite{HOSHI2013-JPSJ}. 
The wavefunction $\phi(\bm{r})$ is drawn by the two isosurfaces painted green and yellow. The isovalues are $\phi(\bm{r}) = \pm C$ ($C>0$).
The wavefunction is  \lq intermediate', because it exhibits intermediate properties between extended and localized wavefunctions.  
}
\label{FIG-NPD-WFN}       
\end{figure}

\section{Summary and overview \label{SEC-SUMMARY}}

The present paper proposes
an a posteriori verification method for the generalized eigenvalue problems that appear in large-scale electronic state calculations. The verification procedure gives a rigorous mathematical foundation of numerical reliability.
In particular, 
the present result guarantees that all of the approximate eigenvalues $\{ \hat{\lambda}_k \}_k$ are well separated and that the participation ratio value $\{ P(\hat{x}_k) \}_k$ 
and any physical quantity defined for each eigenvector are meaningful.  
Since the verification procedure consists of simple matrix multiplications, 
the computational cost is moderate,
as compared with that of the solver procedure.
Therefore, application researchers can use the verification function with only a moderate increase of the computational cost. 
Test calculations were carried out on the K computer for real problems with a matrix size of up to $n \approx 4 \times 10^5$.

The next stage of research is the integration
of the present verifier routine 
and solver routines in EigenKernel,
in which we can use various solver routines among ScaLAPACK and newer libraries and can compare their approximate solutions in the verification procedure. 

Future issues are 
realizing (i) the verification of eigenvectors and (ii) the refinement of approximate eigenpairs. 
The refinement procedure will be crucial, in particular, 
when lower-precision arithmetic, such as half-precision or single-precision arithmetic, is used for calculating an approximate solution as an initial guess.
For example, refinement algorithms for the symmetric eigenvalue problem have recently been proposed in \cite{OgAi2018,OgAi2019}, which are based on matrix multiplications.
Such refinement algorithms enhance application researchers to use lower-precision arithmetic with satisfactory reliability of the computed results, which will be of great importance in next-generation architecture that is optimized for lower-precision arithmetic.

\section*{Acknowledgement}
The authors wish to thank the anonymous referees for their valuable comments, which helped to improve our paper significantly.
The present study was supported in part by 
MEXT as Exploratory Issue 1-2 of the Post-K (Fugaku) computer project ``Development of verified numerical computations and super high-performance computing environment for extreme researches'' using computational resources of the K computer provided by the RIKEN R-CCS through the HPCI System Research project (Project ID: hp180222) and
Priority Issue 7 of the Post-K computer project
and by JSPS KAKENHI Grant Numbers 16KT0016, 17H02828, and 19H04125.

\appendix
\section{Generalized eigenvalue problems in electronic state calculations \label{SEC-GHEV}}

This section introduces a generalized eigenvalue problem
as a numerical foundation of large-scale electronic state calculations.
Details can be found in textbooks, such as Ref.~\cite{MARTIN2004-TEXT}. 
The fundamental Schr\"{o}dinger-type equation, which is a linear partial differential equation, 
is written for an electronic wave function $\phi(\bm{r})$ in real space for a position vector $\bm{r}=(x,y,z)$ as 
\begin{eqnarray}
 H \phi(\bm{r}) =\lambda \phi(\bm{r})
 \label{EQ-QM-EQN}
\end{eqnarray}
with the Hamilton operator
\begin{eqnarray}
 H \equiv - \frac{\hbar^2}{2m} \Delta + V_{\rm eff}(\bm{r}).
\end{eqnarray}
Here, 
$\Delta = \partial_x^2 + \partial_y^2 + \partial_z^2 $ is the Laplacian, 
$m$ is the mass of the electron, $\hbar$ ($\approx 1.05 \time 10^{-34}$ Js) is the Planck constant, and $V_{\rm eff}(\bm{r})$ is the effective potential, which is a scalar function. 
The normalization condition
\begin{eqnarray}
\int |\phi(\bm{r})|^2 = 1  
 \label{EQ-NORMALIZAION}
\end{eqnarray}
is imposed and stems from the fact that the sum of the weight distribution of one electron 
should be unity.  

An eigenvalue $\lambda$ that indicates the energy of an electron in the material is called an eigenenergy. 
The $k$-th eigenpair $(\lambda_k, \phi_k(\bm{r}))$ is defined for $k=1,2,..,n$
in the order of $\lambda_1 \le \lambda_2 \le \cdots \le \lambda_n$.

Now, we consider as a typical case that  
$\phi(\bm{r})$ is expressed as a linear combination of given basic functions 
\begin{eqnarray}
 \phi(\bm{r}) = \sum_{j}^{n} c_j \chi_j(\bm{r}) 
 \label{EQ-QM-LCAO},
\end{eqnarray}
where the basis functions $\{ \chi_j(\bm{r}) \}$ are normalized to be
\begin{eqnarray}
\int \chi^\ast_j(\bm{r})  \chi_j(\bm{r}) d\bm{r} =1.
\end{eqnarray}
A typical basis function is called atomic orbital
and is localized near the position of an atomic nucleus. 
Since each basis function belongs to one atom,
the basis index $i$ is equivalent to the composite indices
of an atom index $\mathcal{I}$ and an orbital index $\alpha$ ($i \equiv (\mathcal{I}, \alpha)$).
The orbital index $\alpha$ distinguishes the basis functions that belong to the same atom but differ in shape. 
Usually, the number of basis functions $n$ is nearly proportional to that of atoms $n_{\rm atom}$
($n \propto n_{\rm atom}$).

When (\ref{EQ-QM-LCAO}) is used for (\ref{EQ-QM-EQN}),
the generalized eigenvalue problem (\ref{EQ-QM-GEP}) appears 
with the real-symmetric $n \times n$ matrices $A$ and $B$, where
\begin{eqnarray}
 A_{ij} &\equiv& \int \chi_i^\ast(\bm{r}) H  \chi_j(\bm{r}) d\bm{r} \\
 B_{ij} &\equiv& \int \chi_i^\ast(\bm{r})  \chi_j(\bm{r}) d\bm{r}. \label{EQ-B-MAT-INTEGRAL}
\end{eqnarray}
The matrix $B$ is positive definite and 
satisfies $B_{jj}=1$ and $| B_{ij}| < 1$($i \ne j$).
Hereafter, we consider 
that the basis functions are real and the matrices $A$ and $B$ are real symmetric. The eigenvectors $x_k$ are real, and the normalization condition of (\ref{EQ-NORMALIZAION}) is reduced to 
\begin{eqnarray}
x_k^{\rm T} B x_k = 1,
\label{EQ-B-NORMALIZATION-V}
\end{eqnarray}
which is called $B$-normalization.

Here, the simplest theory of the hydrogen molecule (H$_2$) is demonstrated,
as in many textbooks, such as Ref.~\cite{ATKINS}.
The atomic nuclei of the first and second hydrogen atoms
are located at  $\bm{r} = \bm{R}_1$ and $\bm{R}_2$, respectively. 
We consider a given localized function $f(\bm{r})$ with  
localization center located at $\bm{r}=0$.
Two basis functions $\chi_1(\bm{r})$ and $\chi_2(\bm{r})$ are given as 
\begin{eqnarray}
\chi_1(\bm{r}) \equiv f(\bm{r}-\bm{R}_1), \quad 
\chi_2(\bm{r}) \equiv f(\bm{r}-\bm{R}_2).
\end{eqnarray}
The generalized eigenvalue problem of 
(\ref{EQ-QM-GEP}) appears with
the 2 $\times$ 2 real symmetric matrices of 
\begin{eqnarray}
A \equiv 
\begin{pmatrix}
a & -t \\
-t & a 
\end{pmatrix}, 
\quad 
B \equiv 
\begin{pmatrix}
1 & s \\
s & 1 
\end{pmatrix}.
\end{eqnarray}
Now, we consider a typical case in which $a,t,s$ are  positive real numbers and $s<1$. 
The off-diagonal element of $t$ or $s$ is the function of the interatomic distance $d \equiv |\bm{R}_1 - \bm{R}_2| $ ($t=t(d), s=s(d)$). 
The eigenvalues are  
\begin{eqnarray}
\lambda _1 \equiv \frac{a -t}{1+s}, \quad \lambda _2 \equiv \frac{a+t}{1-s}
\end{eqnarray}
and the eigenvectors are 
\begin{eqnarray}
x_1 = 
\frac{1}{\sqrt{2(1+s)}}
\begin{pmatrix}
1  \\
1  
\end{pmatrix}, \quad
x_2 = 
\frac{1}{\sqrt{2(1-s)}}
\begin{pmatrix}
1  \\
-1  
\end{pmatrix}.
\end{eqnarray}
The matrix $B$ has the eigenvalues $1 \pm s$
and will be not positive definite
in the limiting situation of $s \rightarrow 1$.
One may suspect that
the limiting situation can appear, 
when the distance between the nuclei of atoms is approximately zero 
$(d \rightarrow 0)$ and that
the two basis functions will be identical
($\chi_1(\bm{r}) - \chi_2(\bm{r}) \rightarrow 0$).
Among real materials, fortunately,
the distance $d$ is so large that the limiting situation does not occur. 

In general, 
the matrix size $n$ in generalized eigenvalue problem (\ref{EQ-QM-GEP})
is nearly proportional to that of the atoms 
($n \propto n_{\rm atom}$).  
For example, 
a typical model for the benzene molecule (C$_6$H$_6$), called the \lq sp' model, 
gives one basis function for each hydrogen atom (H) and 
four basis functions for each carbon atom (C). 
The total number of basis functions or the matrix size $n$
is $n = 1 \times 6 + 4 \times 6 = 30$.

Matrix data of $A$ and $B$ for various materials are stored 
in the ELSES matrix library
\cite{ELSES-MATRIX-LIBRARY-URL, HOSHI2018-PENTA}.
The matrix data were generated 
by the electronic-state calculation software ELSES \cite{ELSES-URL, HOSHI2012-ELSES, HOSHI2016-SC16}
with first-principles-based modeled (tight-binding) electronic-state theory. 
The atomic unit is used for the energy among the data files. 
For example, 
the matrix problem for a benzene molecule (C$_6$H$_6$) in the above model 
is stored as \lq BNZ30'. 
For many problems,
the approximate eigenvalues 
$\{ \ap{\lambda}_k \}$  are uploaded,
as well as the matrix data of $A$ and $B$, for the convenience of the researcher.
The sparsity of the stored matrix data of $A_{ij}$ and $B_{ij}$ is explained briefly.
As explained above, the indices $i$ and $j$ are the composite indices of the atom indices $\mathcal{I}$ 
and $J$ and the orbital indices $\alpha$ and $\beta$, respectively ($i \equiv i(I,\alpha), j \equiv j(J,\beta))$).
Therefore, an element of the matrices $A$ and $B$ is expressed by the four indices as $A_{I\alpha;J\beta}$ and $B_{I\alpha;J\beta}$, respectively.
Since a matrix element value decreases quickly and monotonically 
as a function of the inter-atomic distance between the $I$-th and $J$-th atoms ($r_{IJ}$), a cutoff distance $r_{\rm cut}$ can be introduced.
A matrix element, $A_{I\alpha;J\beta}$ or $B_{I\alpha;J\beta}$, is ignored, if $r_{IJ} > r_{\rm cut}$, which makes the matrices sparse.
More information on the data file in the ELSES matrix library is found in Ref.~\cite{HOSHI2018-PENTA}.

As a future issue,
we should consider the numerical error 
in the input matrix data of $A$ and $B$,
because this error may affect the final conclusion.
The possible numerical error can be decomposed into the two terms 
\begin{eqnarray}
A_{\rm exact} - A = \delta A_{\rm cut} + \delta A_{\rm cal} \\
B_{\rm exact} - B = \delta B_{\rm cut} + \delta B_{\rm cal}, 
\end{eqnarray}
where $A_{\rm exact}$ and $B_{\rm exact}$ are the exact (theoretical) matrix data.
The error terms $\delta A_{\rm cut}$ and $\delta B_{\rm cut}$ are called cutoff errors 
and stem from the cutoff procedure explained in the previous paragraph.
The maximum element of $\delta A_{\rm cut}$ or $\delta B_{\rm cut}$ is on the order of $10^{-4}$ in the case of PPE17994, for example. 
The cutoff error term will be eliminated
when the full matrix data are adopted in the input matrices.  
The full matrix data are not stored in the ELSES matrix library, because these data consume a large amount of disk space.
We intend to perform verification using the full matrix data when we integrate the verifier routines into the simulation software (ELSES), which generates the matrix data and includes the solver routine. 
The procedure with the full matrix data will not increase the computational cost, 
because all of the procedures use the dense-matrix algorithms. 
The error terms $\delta A_{\rm cal}$ and $\delta B_{\rm cal}$, on the other hand,   
are called calculation errors and stem from the generating procedure of the matrices. 
The matrix $B$ is defined as the integral in (\ref{EQ-B-MAT-INTEGRAL}). 
The three-dimensional numerical integral of (\ref{EQ-B-MAT-INTEGRAL}) is obtained for the Slater-type functions $\{ \chi_i (\bm{r}) \}_i$ \cite{SLATER1930}
in prolate spheroidal coordinates,
which is reviewed in Section II of Ref.~\cite{LESIUK2014}.
The matrix $A$ is calculated from the matrix $B$ 
in a modeled (atom superposition and electron delocalization tight-binding) theory 
\cite{NATH1990, CALZAFFERI1996, HOSHI2012-ELSES}.
Evaluating the calculation errors $\delta A_{\rm cal}$ and $\delta B_{\rm cal}$ is an interesting topic and will be discussed in the near future.

\bibliographystyle{elsarticle-num} 
\bibliography{bibtex_list}

\end{document}